\documentclass[12pt]{article}

\usepackage{epsfig}
\usepackage{subfigure}
\usepackage{calc}
\usepackage{amssymb}
\usepackage{amstext}
\usepackage{amsmath}
\usepackage{multicol}
\usepackage{pslatex}
\usepackage[small]{caption}

\usepackage{MnSymbol}
\usepackage{color}
\usepackage{balance}
\usepackage{latexsym}
\usepackage{graphicx}
\usepackage{algorithm}
\usepackage[noend]{algorithmic}
\usepackage[T1]{fontenc}
\usepackage{epic}
\usepackage{eepic}
 \usepackage{amsfonts}
\usepackage{stmaryrd}
\usepackage[all]{xy}
\usepackage{booktabs}
\usepackage{multirow}
\usepackage{rotating}

\newtheorem{theorem}{Theorem}

\newtheorem{corollary}[theorem]{Corollary}
\newtheorem{definition}{Definition}
\newtheorem{example}{Example}

\algsetup{linenodelimiter=.} 
\algsetup{indent=1em}

\definecolor{light-gray}{gray}{0.90}

\newcounter{theorem-backup}

 \newtheorem{proof}{Proof}

\newcommand{\Z}{\mathbb{Z}}
\newcommand{\R}{\mathbb{R}}

\newcommand{\B}{\mathbb{B}}

\usepackage{listings}
\lstdefinelanguage{PseudoC}[ISO]{C++} { 
morekeywords={foreach, and, not, or, is, FIFO_Queue, HashTable, FILE, Cache}, 
morekeywords=[2]{HashInsert, Enqueue, Dequeue, next}, 
mathescape=true,
alsoletter={'},
deletestring=[b]'
}

\lstdefinelanguage{PseudoCWithNumbers}[ISO]{C++} { 
morekeywords={foreach, and, not, or, is, FIFO_Queue, HashTable, FILE, Cache}, 
morekeywords=[2]{HashInsert, Enqueue, Dequeue, next}, 
mathescape=true,
alsoletter={'},
deletestring=[b]'
}

\lstdefinelanguage{Murphi}[]{Pascal} { 
morekeywords={ruleset, rule, invariant, startstate, return, endif, endfor, endswitch, forall, endforall, exists, endexists}, 
mathescape=true, 
morestring=[b]",
morestring=[b]', 
morecomment=[s]{/*}{*/} ,
morecomment=[l]{--}
}

\lstdefinelanguage{PRISM}[ISO]{C++} { 
morekeywords={probabilistic, stochastic, const, rate, module, endmodule, init, P, U},
mathescape=true, 
alsoletter={'},
deletestring=[b]'
}

\lstdefinelanguage{Yacc}[ISO]{C++} { 
morekeywords={token, left}, 
mathescape=false,
alsoletter={'},
deletestring=[b]'
}

\lstdefinestyle{PseudoC}{
language=PseudoC,
basicstyle=\ttfamily,
tabsize=1,
showlines=false,
emptylines=*1,
breaklines=true,
breakindent=5pt,
keywordstyle=\rmfamily\bfseries,
keywordstyle=[2]\rmfamily,
commentstyle=\itshape, 
columns=fixed,
showspaces=false, 
showstringspaces=false, 
showtabs=false, 
escapechar=\%
}

\lstdefinestyle{PseudoCWithNumbers}{
language=PseudoC,
basicstyle=\ttfamily,
tabsize=1,
showlines=false,
emptylines=*1,
breaklines=true,
breakindent=5pt,
keywordstyle=\rmfamily\bfseries,
keywordstyle=[2]\rmfamily,
commentstyle=\itshape, 
columns=fixed,
showspaces=false, 
showstringspaces=false, 
showtabs=false, 
escapechar=\%,
numbersep=5pt,framexleftmargin=15pt,numbers=left,
}

\lstdefinestyle{Murphi}{
language=Murphi,
basicstyle=\ttfamily,
tabsize=1,
showlines=false,
emptylines=*1,
breaklines=true,
breakindent=5pt,
keywordstyle=\rmfamily\bfseries,
commentstyle=\itshape, 
columns=fixed,
showspaces=false, 
showstringspaces=false, 
showtabs=false,
escapechar=\%}

\lstdefinestyle{PRISM}{
language=PRISM,
basicstyle=\ttfamily,
tabsize=1,
showlines=false,
emptylines=*1,
breaklines=true,
breakindent=5pt,
keywordstyle=\rmfamily\bfseries,
commentstyle=\itshape, 
columns=fixed,
showspaces=false, 
showstringspaces=false, 
showtabs=false,
escapechar=\%}

\lstdefinestyle{Yacc}{
language=PseudoC,
basicstyle=\ttfamily,
tabsize=1,
showlines=false,
emptylines=*1,
breaklines=true,
breakindent=5pt,
keywordstyle=\rmfamily\bfseries,
keywordstyle=[2]\rmfamily,
commentstyle=\itshape, 
columns=fixed,
showspaces=false, 
showstringspaces=false, 
showtabs=false, 
}

\newcommand{\qks}{\mbox{QKS}}

\renewcommand{\eqref}[1]{(\ref{#1})}

\definecolor{Blue}{rgb}{0.25,0.33,0.77}
\definecolor{Red}{rgb}{1,0,0}
\definecolor{Green}{rgb}{0,1,0}

\title{\Large \bf
Automatic Control Software Synthesis for Quantized Discrete Time Hybrid Systems
}

\author{Vadim Alimghuzin, Federico Mari, Igor Melatti,\\
Ivano Salvo, Enrico Tronci\\
\small \itshape Department of Computer Science\\
\small \itshape Sapienza University of Rome\\
\small \itshape via Salaria 113, 00198 Rome\\
\small email: \{alimghuzin,mari,melatti,salvo,tronci\}@di.uniroma1.it\\
\small \itshape Vadim Alimghuzin is also with the Department of Computer Science and Robotics\\
\small \itshape Ufa State Aviation Technical University\\
\small \itshape 12 Karl Marx Street, Ufa, 450000, Russian Federation}
\begin{document}

  \maketitle  


\begin{abstract}
Many \emph{Embedded Systems} are indeed \emph{Software Based Control Systems},
that is control systems whose controller consists of \emph{control software} running on a 
microcontroller device. 
This motivates investigation on \emph{Formal Model Based Design} 
approaches for automatic synthesis of embedded systems control software. 
This paper addresses control software synthesis for discrete time {\em nonlinear} systems.  
We present a methodology to overapproximate the dynamics of a discrete time 
nonlinear hybrid system ${\cal H}$ by means of a discrete time {\em linear} hybrid system
${\cal L}_{\cal H}$, in such a way that controllers for ${\cal L}_{\cal H}$ are guaranteed to be
controllers for ${\cal H}$.  
We present experimental results on 
the inverted pendulum, a challenging and meaningful benchmark in 
nonlinear Hybrid Systems control. 
\end{abstract}


 
  \section{Introduction}
\label{intro.tex}




Many \emph{Embedded Systems} are indeed \emph{Software Based Control Systems} (SBCSs).
An SBCS 
consists of two main subsystems:
the \emph{controller} and the \emph{plant}.
Typically, the plant is a physical system
consisting, for example, of mechanical or electrical devices
whereas the controller 
consists of \emph{control software} running on a microcontroller.
In an endless loop, 
the controller
reads \emph{sensor} outputs from the plant and
sends commands to plant \emph{actuators}
in order to guarantee that the 
\emph{closed loop system}
(that is, the system consisting of both
plant and controller) meets given
\emph{safety} and \emph{liveness} specifications
(\emph{System Level Formal Specifications}).

Software generation from models
and formal specifications forms the core of
\emph{Model Based Design} of embedded software 
\cite{Henzinger-Sifakis-fm06}.
This approach is particularly interesting for SBCSs
since in such a case system level (formal) specifications are much easier
to define than the control software behavior itself.


%

The typical control loop skeleton for an SBCS is the following.
Measure $x$ of the system state from plant \emph{sensors} go through an
\emph{analog-to-digital} (AD) conversion, yielding a {\em quantized} value $\hat{x}$. 
A function \texttt{ctrlRegion} checks if $\hat{x}$ belongs to the region in which 
the control software works correctly. If this is not the case a \emph{Fault Isolation and 
Recovery} (FDIR) procedure is triggered, otherwise a function \texttt{ctrlLaw} computes a 
command $\hat{u}$ to be sent to plant \emph{actuators} after a \emph{digital-to-analog} (DA) 
conversion.
Basically, the control software design problem for SBCSs consists
in designing software implementing functions
\texttt{ctrlLaw} and
\texttt{ctrlRegion}.



For SBCSs, system level specifications are typically
given with respect to the desired behavior 
of the closed loop system. The 
\emph{control software} 
(that is, \texttt{ctrlLaw} and \texttt{ctrlRegion})
is designed using a 
\emph{separation-of-concerns} approach. That is, 
\emph{Control Engineering} techniques 
(e.g., see \cite{modern-control-theory-1990})
are used
to design, from the closed loop system level specifications,
\emph{functional specifications} (\emph{control law})
for the \emph{control software}
whereas
\emph{Software Engineering} techniques are used to
design control software implementing the given
functional specifications.
Such a separation-of-concerns approach has 
several
drawbacks. 

First, usually control engineering techniques do not
yield a formally verified specification for the
control law 
when quantization is taken into account. 
This is particularly the case when the plant has to be modelled 
as a \emph{Hybrid System}, 
that is a system with continuous as well as discrete state changes
\cite{alur-sfm04,alg-hs-tcs95,HHT96,AHH96}.
As a result, even if the control software meets its functional
specifications there is no formal guarantee that
system level specifications are met since
quantization effects are not formally accounted for.

Second, issues concerning computational resources, 
such as control software \emph{Worst Case Execution Time} (WCET),
can only be considered very late in the SBCS design activity, namely once the 
software has been designed. As a result, the control software may have a WCET
greater than the sampling time. 
This invalidates the schedulability
analysis (typically carried out before the control software is completed)
and may trigger redesign of the software or even 
of its functional specifications
(in order to simplify its design).

Last, but not least, the classical 
separation-of-concerns approach does not effectively support 
design space exploration for the control software. 
In fact, although in general there will be many functional specifications
for the control software that will allow meeting the given
system level specifications, the software engineer only
gets one 
 to play with. This overconstrains a priori 
 the design space
for the control software implementation preventing, for example, effective performance trading 
(e.g., 
between number of bits in AD conversion,  
WCET, 
RAM usage, 
CPU power consumption, 
etc.).
We note that the above considerations also apply
to the typical situation where
Control Engineering techniques are used to design a control
law and then tools like Simulink are used to generate the control
software. 


The previous considerations motivate research on 
Software Engineering methods and tools focusing on control software synthesis
(rather than on control law synthesis as in Control Engineering).
The objective is that from the plant model (as a hybrid system), 
from formal specifications for the closed loop system behavior 
and from \emph{Implementation Specifications} 
(that is, number of bits used in the quantization process)
such methods and tools can generate correct-by-construction control software
satisfying the given specifications.

The tool QKS \cite{qks-cav10} synthesise control software for 
{\em Discrete Time Linear Hybrid Systems} (DTLHSs). 
However, the dynamics of many interesting hybrid systems 
cannot be directly modeled by linear predicates. 
The focus of the present paper is control software 
synthesis for {\em nonlinear} Discrete Time Hybrid Systems. 


\subsection{Our Main Contributions}

We model the controlled system (plant) as a
\emph{Discrete Time Hybrid System} 
(DTHS), that is a discrete time hybrid system
whose dynamics is modeled as a {\em predicate} (possibly non linear)  
over a set of continuous as well as discrete 
variables that describe system state, system inputs and disturbances.

System level safety as well as liveness specifications are modeled as sets of states
defined, in turn, as predicates. 
In our setting, as always in control problems, liveness constraints define the
set of states that any evolution of the closed loop system should eventually reach
(\emph{goal states}). Using an approach similar to the one in
 \cite{decidability-hybrid-automata-jcss98}, in~\cite{cav2010-tekrep-art-2011} 
 it has been proven that both existence of a controller  
 and existence of a \emph{quantized} controller for DTHSs are undecidable problems, 
 even for very restricted classes of DTHSs.
Accordingly, we can only hope for non complete or semi-algorithms.

In this paper we present a general approach to deal with discrete time non-linear hybrid systems.
The basic idea is 
to overapproximate the behaviour of a DTHS ${\cal H}$ by means of a 
DTLHS ${\cal L}_{\cal H}$.
Stemming from 
Corollary~\ref{cor:real-result},
that ensures that  
controllers for ${\cal L}_{\cal H}$ are guaranteed to be controllers for ${\cal H}$,  
we synthesize control software by giving as input to 
the tool QKS \cite{qks-cav10} the linear plant model ${\cal L}_{\cal H}$, the desired 
quantization schema, and system level formal specifications.

Since ${\cal L}_{\cal H}$ dynamics overapproximates the dynamics of ${\cal H}$, 
the controllers that we synthesize are inherently {\em robust}, that is they meet 
the given closed loop requirements 
{\em notwithstanding} nondeterministic small \emph{disturbances} such as
variations in the plant parameters.
Tighter overapproximations makes finding a controller easier, whereas 
coarser overapproximations makes controllers more robust.
As in the linear case, the automatically generated software 
has a \emph{Worst Case Execution Time} (WCET)
guaranteed to be  linear in the number of bits of the state quantization schema. 
Moreover, control software computes commands in such a way that the closed 
loop system follows a (near) {\em time optimal} strategy to reach 
the goal~\cite{Girard:2010gf}.  
We present
experimental results 
on the inverted pendulum benchmark~\cite{KB94}, a challenging and 
well studied example in control synthesis.

\subsection{Related Work}
\label{related-works.tex}



%
%
%
%
%
Control Engineering has been studying
control law design (e.g., optimal control, robust control, etc.),
for more than half a century
(e.g., see~\cite{modern-control-theory-1990}).
Also \emph{Quantized Feedback Control} 
has been widely studied in control engineering
(e.g. see \cite{quantized-ctr-tac05}).
However such research does not address hybrid systems
(our case) and, as explained above, focuses on control law design rather
than on control software synthesis (our goal).
Furthermore, all control engineering approaches
model \emph{quantization errors} as statistical \emph{noise}. 
As a result, correctness of the control law holds in a probabilistic sense.
Here instead, we model quantization errors as nondeterministic 
(\emph{malicious})
\emph{disturbances}. This guarantees system level correctness of the generated
control software (not just that of the control law) 
with respect to \emph{any} possible sequence of quantization errors.

When the plant model is a \emph{Linear Hybrid Automaton} (LHA)
\cite{alg-hs-tcs95,AHH96}  
reachability and existence of a control law are both undecidable problems
\cite{dt-ctr-rect-aut-icalp97,decidability-hybrid-automata-jcss98}.
This, of course, has not prevented devising effective 
(semi) algorithms for such problems. Examples are in
\cite{AHH96,HHT96,phaver-sttt08,ctr-lha-cdc97,TLS99,Benerecetti:2011yq}.
Control software synthesis for continuous time linear systems
(no switching) has been implemented in the tool 
{\sc Pessoa} \cite{pessoa-cav10}. 
Such an approach exploits suitable finite state abstraction 
(e.g. see \cite{tabuada-quantized-hscc07,pola-girard-tabuada-2008})
to synthesize a control law computing commands
from real valued state measures (no quantization). 
The control software is then generated by passing to 
Simulink such a control law. 
In the same wavelength, \cite{Yordanov:2012ul} generates a
control strategy from a finite abstraction of a  
\emph{Piecewise Affine Discrete Time Hybrid Systems} (PWA-DTHS). 
Also the Hybrid Toolbox \cite{HybTBX} considers PWA-DTHS. 
Such a tool outputs a feedback control law that is 
then passed to Matlab in order to generate control software.
Finite horizon control of PWA-DTHS
has been studied using
a MILP based approach.
See, for example, 
\cite{sat-opt-ctr-hscc04}.
%
Explicit 
finite horizon control synthesis algorithms for discrete time 
(possibly non-linear) hybrid systems have been studied in 
\cite{icar08} and citations thereof.

We note that all such approaches do not account for state feedback 
quantization since they all assume \emph{exact} (i.e. real valued) state measures. 
Thus, as explained above, they do not
offer any formal guarantee about system level correctness 
of the generated software, which is instead our focus here.

Quantization can be seen as a sort of abstraction,
which has been widely studied in a 
hybrid system formal verification context
(e.g., see \cite{AHLP-ieee00,pred-abs-tecs06}).
Note however that in a verification context abstractions
are designed so as to ease the verification task whereas
in control software synthesis 
quantization is a design requirement since it models a hardware component
(AD converter) which is part of the specification of the control software synthesis problem.
Indeed, in our setting, we have to design a controller \emph{notwithstanding} the nondeterminism stemming
from the quantization process. As a result, the techniques used to devise clever abstractions
in a verification setting cannot be directly used in our synthesis setting where quantization is given.

The tool \qks ~\cite{qks-cav10} synthesize control software from 
system level specification for Discrete Time Linear Hybrid Systems whenever a
a constructive sufficient condition for control software existence holds.
Here, we address control software synthesis for a more general class 
of discrete time hybrid systems.  

In the context of Hybrid Systems verification, the overapproximation of Hybrid 
Systems with Linear Hybrid Systems has been studied in \cite{hypertech-hscc00} and 
\cite{HenzingerH95}. Such works consider dense time models, and focus on 
verification rather than control synthesis. Moreover, we observe that 
we can obtain tighter approximations, since DTLHSs allow us to model system dynamics 
with predicates that mix present and next state variables.

Correct-by-construction software synthesis
in a finite state setting
has been studied, for example, in
\cite{emerson-toplas-04,Tro98,strong-planning-98}.
Such approaches cannot be directly used in our context since they
cannot handle 
continuous state variables.

Summing up, to the best of our knowledge, no previously published
result is available about automatic generation
of correct-by-construction control software
from a DTHS model of the plant, 
\emph{system level formal specifications} and
\emph{implementation specifications} ({\em quantization}, that is number of bits in AD conversion).
\section{Background}
\label{basic.tex}
\label{vars-notation.lab}

We denote with $[n]$ an initial segment $\{1,\ldots, n\}$ of the natural numbers. 
We denote with $X$ = $[x_1, \ldots, x_n]$ a
finite sequence (list) of variables.
By abuse of language we may regard sequences as sets and
we use $\cup$ to denote list concatenation.
Each variable $x$ ranges on a known (bounded or unbounded)
interval ${\cal D}_x$ either of the reals or of the integers (discrete
variables). 
We denote with ${\cal D}_X$ the set $\prod_{x\in X} {\cal D}_x$.
To clarify that a variable $x$ is {\em continuous} 
(i.e. real valued) we may write $x^{r}$. Similarly, 
to clarify that a variable $x$ is {\em discrete} 
(i.e. integer valued) we may write $x^{d}$. 
Analogously $X^{r}$ ($X^{d}$) denotes the sequence
of real (integer) variables in $X$.
Finally, boolean variables are discrete variables
ranging on the set $\B$ = \{0, 1\}. 
If $x$ is a boolean variable 
we write $\bar{x}$ for $(1 - x)$.

\subsection{Predicates}
\label{subsection:predicates}

\sloppy

An {\em expression} $E(X)$ over a list of variables $X$ is an expression of the form 
$\sum_{i\in[n]} a_i f_i(X)$, where 
$f_i(X)$ is a possibly nonlinear function over $X$ 
and $a_i$ are rational constants.
$E(X)$ is a {\em linear expression} 
if each $f_i(X)$ is a projection (i.e. $f_i(X)=x_i$), 
i.e. if it is a linear combination of variables $\sum_{i\in[n]}a_i x_i$. 
A constraint is an expression of the form $E(X)\leq b$, 
where $b$ is a rational constant.
In the following, we also write $E(X) \geq b$
for $-E(X) \leq -b$.

\fussy

{\em Predicates} are inductively defined as follows.
A {\em constraint} $C(X)$ over a list of variables $X$ is a predicate over 
$X$. 
If $A(X)$ and $B(X)$ are predicates over $X$, then $(A(X) \land B(X))$
and $(A(X) \lor B(X))$ are predicates over X.  Parentheses may be
omitted, assuming usual associativity and precedence rules of logical
operators.
A {\em conjunctive predicate} is a conjunction of constraints.
For conjunctive predicates we will also write: 
$E(X) = b$ for (($E(X) \leq b$) $\wedge$ ($E(X) \geq b$)) and $a \leq x \leq
b$ for $x \geq a \;\land\; x \leq
b$, where $x \in X$.
%


A {\em valuation} over a list of variables $X$
is a function $v$ that maps each variable $x \in X$ to a value $v(x)
\in {\cal D}_x$.
Given a valuation $v$, we denote with $X^\ast\in {\cal D}_X$ the sequence of values 
$[v(x_1),\ldots,v(x_n)]$. By abuse of language, 
we call valuation also the sequence of values $X^\ast$.
A \emph{satisfying assignment} to a predicate $P$ over $X$ is a
valuation $X^{*}$ such that $P(X^{*})$ holds. If a satisfying assignment to a
predicate $P$ over $X$ exists, we say that $P$ is {\em feasible}.
Abusing notation, we may
denote with $P$ the set of satisfying assignments to the predicate 
$P(X)$. Two predicates $P$ and $Q$ over $X$ are {\em equivalent},
denoted by $P\equiv Q$, if they have the same set of 
satisfying assignments. Two predicates $P$ and $Q$ are {\em equisatisfiable}
if $P$ is feasible iff $Q$ is feasible.

A variable $x\in X$ 
is said to be {\em bounded} in $P$ if 
there exist $a$, $b \in {\cal D}_x$
such that $P(X)$ implies $a \leq x \leq b$.
A predicate $P$ is bounded if all its variables are bounded.
%
%
%
%
%

Given a constraint $C(X)$ and a fresh boolean variable ({\em guard}) $y \not\in X$,
the {\em guarded constraint} $y \to C(X)$ (if $y$ then $C(X)$) denotes
the predicate $((y = 0) \lor C(X))$. Similarly, we use $\bar{y} \to
C(X)$ (if not $y$ then $C(X)$) to denote the predicate $((y = 1) \lor
C(X))$.
A {\em guarded predicate} is a conjunction of 
either constraints or guarded constraints.
It is possible to show that, 
if a guarded predicate $P$ is bounded,
then $P$ can be transformed into an equivalent (bounded) conjunctive predicate~\cite{cav2010-tekrep-art-2011}. 

\subsection{Labeled Transition Systems}
\label{lts.tex}

A \emph{Labeled Transition System} (LTS) is a tuple
${\cal S} = (S, A, T)$ where 
$S$ is a (possibly infinite) set of states, 
$A$ is a (possibly infinite) set of \emph{actions}, and 
$T$ : $S$ $\times$ $A$ $\times$ $S$ $\rightarrow$ $\B$
is the \emph{transition relation} of ${\cal S}$.
We say that $T$ (and ${\cal S}$) is {\em deterministic} if $T(s, a, s') \land
T(s, a, s'')$ implies $s' = s''$, and {\em nondeterministic}
otherwise. 
Let $s \in S$ and $a \in A$.
%
%
%
%
%
%
We denote with 
$\mbox{\rm Adm}({\cal S}, s)$ the set of actions
admissible in $s$, that is $\mbox{\rm Adm}({\cal S}, s)$ = $\{a \in A
\; | \; \exists s': T(s, a, s') \}$
and with
$\mbox{\rm Img}({\cal S}, s, a)$ the set of next
states from $s$ via $a$, that is $\mbox{\rm Img}({\cal S}, s, a)$ =
$\{s' \in S \; | \; T(s, a, s') \}$.
A {\em run} or \emph{path}
for an LTS ${\cal S}$ 
is a sequence 
$\pi$ =
$s_0, a_0, s_1, a_1, s_2, a_2, \ldots$ 
of states $s_t$ and actions $a_t$ 
such that
$\forall t \geq 0$ $T(s_t, a_t, s_{t+1})$.
The length $|\pi|$ of a finite run $\pi$ is the number of actions
in $\pi$. 
We denote with $\pi^{(S)}(t)$ the $(t + 1)$-th state element of
$\pi$, and with $\pi^{(A)}(t)$ the $(t + 1)$-th action element of
$\pi$. That is $\pi^{(S)}(t)$ = $s_t$, and $\pi^{(A)}(t)$ = $a_t$.

Given two LTSs ${\cal S}_1$ $=$ $(S$, $A$, $T_1)$ and ${\cal S}_2$ $=$ $(S$,
  $A$, $T_2)$, we say that ${\cal S}_1$ \emph{refines} ${\cal
    S}_2$ (notation ${\cal S}_1 \sqsubseteq {\cal S}_2$) iff $T_{1}(s,
  a, s')$ implies $T_{2}(s, a, s')$ for each state $s,s' \in S$ and action $a \in A$.
The refinement relation is a partial order on LTSs.
\subsection{LTS Control Problem}
A \emph{controller} for an LTS ${\cal S}$ 
is used to restrict the dynamics of ${\cal S}$ 
so that all states in the initial region 
will reach in one or more steps the goal region. 
In the following, we formalize such a concept
by defining strong 
solutions to an LTS control problem. 
In what follows, let ${\cal S} = (S, A, T)$ be an LTS, 
$I$, $G$ $\subseteq$ $S$ 
be, respectively, the {\em initial} and {\em goal} regions of ${\cal S}$.

\begin{definition}
  \label{def:ctroller-lts}
  \label{def:ctrproblem-lts}
  A \emph{controller} for 
${\cal S}$ is a function 
  $K : S \times A \rightarrow \B$
  such that $\forall s \in S$, $ \forall a \in A$, if $K(s, a)$ then
  $\exists s' \; T(s, a, s')$.
  %
$\mbox{\rm dom}(K)$ denotes the set of states for which at least a
  control action is enabled. Formally, $\mbox{\rm dom}(K)$ $=$ $\{s \in
  S \; | \; $$\exists a \; K(s, a)\}.$
  %
${\cal S}^{(K)}$ denotes the \emph{closed loop system}, that
  is the LTS $(S, A, T^{(K)})$, where 
$T^{(K)}(s, a, s')$ $=$ $T(s, a, s') \wedge K(s, a)$.
%
%
%
%
%
%
%
\end{definition}
We call a path $\pi$ {\em fullpath} \cite{emerson-toplas-04}
 if either it is infinite or its last state 
$\pi^{(S)}(|\pi|)$ has no successors 
(i.e. $\mbox{\rm Adm}({\cal S}, \pi^{(S)}(|\pi|)) = \varnothing$).
We denote with ${\rm Path}(s, a)$ the set of fullpaths starting in state
$s$ with action $a$, i.e. the set of fullpaths $\pi$ such that $\pi^{(S)}(0)=s$
and $\pi^{(A)}(0)=a$.

Given a path $\pi$ in ${\cal S}$, 
we define $J({\cal S},\pi,G)$ as follows. If there exists $n > 0$ s.t.
$\pi^{(S)}(n)\in G$, then $J({\cal S},\pi,G) = \min\{n \;|\; n >
0 \land \pi^{(S)}(n)\in G\}$. Otherwise, $J({\cal S},\pi,G) = +\infty$.
We require $n > 0$ since
our systems are nonterminating and each controllable state (including a goal state)
must have a path of positive length to a goal state.
Taking ${\rm sup}\, \varnothing = +\infty$ and $\inf\, \varnothing = -\infty$, 
the {\em worst case distance} 
of a state $s$ from the goal region  $G$
is  $J_{\rm strong}({\cal S},G,s)={\rm sup} \{ J_s({\cal S},G,s, a)~|~ a \in {\rm
Adm}({\cal S},s)\}$, being $J_s({\cal S},G,s,a)={\rm sup} \{ J({\cal S},G,\pi)~|~ \pi
\in{\rm Path}(s,a)\}$. 


\begin{definition}
\label{def:sol}
A \emph{control problem} for 
${\cal S}$ 
is a triple 
  $\cal P$ = $({\cal S}, I, G)$. 
A {\em strong} 
solution (or simply a solution) to 
${\cal P}$ 
is 
a controller $K$ for ${\cal S}$, such that $I$ $\subseteq$ $\mbox{\rm dom}(K)$ and
for all $s \in \mbox{\rm Dom}(K)$, 
$J_{strong}({\cal S}^{(K)}, G, s)$ 
is finite.

An \emph{optimal} solution to ${\cal P}$
is a solution $K^{*}$ to ${\cal P}$ s.t.
for all solutions
$K$ to ${\cal P}$, 
for all $s \in {\cal D}_{X}$
 we have: 
$J_{strong}({\cal S}^{(K^{*})}, G, s) \leq J_{strong}({\cal S}^{(K)}, G, s)$.
The  \emph{most general optimal (mgo) solution}  to ${\cal P}$ 
is an optimal solution $\bar{K}$ to ${\cal P}$  s.t. 
for all optimal solutions $K$ to ${\cal P}$,
for all $s \in {\cal D}_{X}$,
for all $u \in {\cal D}_{U}$
we have: $K(s, u)$ $\rightarrow$ $\bar{K}(s, u)$. It is easy to see that
this definition is well posed (i.e., the mgo solution is unique) 
and that $\bar{K}$ does not depend on $I$.
\end{definition}

  \section{Discrete Time Hybrid Systems} \label{dths.tex}

In this section we introduce our class of {\em Discrete Time Hybrid Systems} (DTHS for short), 
together with the DTHS representing the inverted pendulum 
on which our experiments will focus. 
Moreover, we will define in Sect.~\ref{sec:qcp} the {\em Quantized Control Problem}.

 

\begin{definition}
\label{dths.def}

A {\em Discrete Time Hybrid System} 
is a tuple ${\cal H} = (X,$ $U,$ $Y,$ $N)$ where:

\begin{itemize}

\item 
  $X$ = $X^{r} \cup X^{d}$ 
  is a finite sequence of real ($X^{r}$) and 
  discrete ($X^{d}$) 
  {\em present state} variables.  
  We denote with $X'$ the sequence of 
  {\em next state} variables obtained 
  by decorating with $'$ all variables in $X$.
  
\item 
  $U$ = $U^{r} \cup U^{d}$ 
  is a finite sequence of 
  \emph{input} variables.

\item 
  $Y$ = $Y^{r} \cup Y^{d}$ 
  is a finite sequence of
  \emph{auxiliary} variables. 
  Auxiliary variables are typically used to
  model \emph{modes} (e.g., from switching elements such as diodes) 
  or ``local'' variables.

\item 
  $N(X, U, Y, X')$ is a conjunctive predicate 
  over $X \cup U \cup Y \cup X'$ defining the 
  {\em transition relation} (\emph{next state}) of the system.
  $N$ is {\em deterministic} if $N(x, u, y_1, x')$ $\land$ $N(x, u, y_2, x'')$
  implies $x' = x''$, and {\em nondeterministic} otherwise.
\end{itemize}

A DTHS is {\em bounded} if the predicate $N$ is bounded.
A DTHS is {\em deterministic} if $N$ is deterministic.
A DTHS is {\em linear}, and we call it DTLHS if $N$ is a conjunction of linear constraints. 
\end{definition}

Since any bounded guarded predicate can 
be transformed into a conjunctive predicate (see Sect.~\ref{subsection:predicates}), for the sake of readability we will 
use bounded guarded predicates to describe the transition relation of 
bounded DTHSs. To this aim, we will also clarify which variables are boolean,
and thus may be used as guards in guarded constraints.


\begin{example}
\label{ex:dths}

\begin{figure}
  \centering
  \includegraphics[width=0.3\columnwidth]{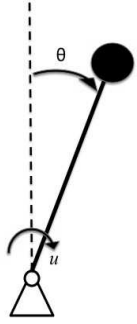}
\vspace*{-0.4cm}
  \caption{Inverted Pendulum with Stationary Pivot Point.}
  \label{fig:invpend}
\vspace*{-0.6cm}
\end{figure}

Let us consider a simple inverted pendulum~\cite{KB94}, as shown in Fig.~\ref{fig:invpend}. 
The system is modeled by taking the angle $\theta$ and the angular velocity $\dot{\theta}$ as 
state variables. The input of the system is the torquing force $u$, that can influence 
the velocity in both directions. 
Moreover, the behaviour of the system depends on the pendulum mass $m$, the length of the 
pendulum $l$ and the gravitational acceleration $g$. Given such parameters, 
the motion of the system is described by the differential equation 
$\ddot{\theta} = {g \over l} \sin \theta + {1 \over m l^2} u$.

In order to obtain a state space representation, we consider the following normalized system,  
where $x_1$ is the angle $\theta$ and $x_2$ is the angular speed $\dot{\theta}$.

\begin{equation}
  \label{eq:pendmotion2}
\left\{\begin{array}{l}
    \dot{x}_1 = x_2 \\
    \dot{x}_2 = {g \over l} \sin x_1 + {1 \over m l^2} u
  \end{array}
\right.
\end{equation}

The DTHS model ${\cal H}$ for the pendulum is the tuple $(X,U,Y,N)$, 
where $X=\{x_1, x_2\}$ is the set of continuous state variables, 
$U=\{u\}$ is the set of input variables, and $Y=\varnothing$. 
Differently from~\cite{KB94}, we consider the problem of finding a discrete controller, whose decisions maybe ``apply the force clockwise'' ($u=1$), ``apply the force counterclockwise'' ($u=-1$)'', or ``do nothing'' ($u=0$). 
The intensity of the force will be given as a constant $F$. Finally, the discrete time 
transition relation $N$ is obtained from the equations in (\ref{eq:pendmotion2}) by introducing
a constant $T$ that models the sampling time. $N$ is the 
predicate $(x'_1 = x_1 + T x_2)\,\land\,(x'_2 = x_2 + T {g \over l} \sin x_1 + T {1 \over m l^2} F u)
$. 
\end{example}

The semantics of DTHSs is given in terms of LTSs. 

\begin{definition}
Let ${\cal H}$ = ($X$, $U$, $Y$, $N$) 
be a DTHS.
The dynamics of ${\cal H}$ 
is defined by the Labeled Transition System 
$\mbox{\rm LTS}({\cal H})$ = (${\cal D}_X$, ${\cal D}_U$,
$\tilde{N}$) where:
$\tilde{N} : {\cal D}_X \; \times \; {\cal D}_U \; \times \; {\cal D}_X \rightarrow \B$ 
is a function s.t.  $\tilde{N}(x, u, x') \equiv \exists y \in {\cal D}_Y: N(x, u, y, x')$.
A \emph{state} $x$ for ${\cal H}$ is a state $x$ for 
$\mbox{\rm LTS}({\cal H})$ and a \emph{run} 
(or \emph{path}) for ${\cal H}$ is
a run for $\mbox{\rm LTS}({\cal H})$ (Sect. \ref{lts.tex}).
\end{definition}

\subsection{DTHS Control Problem}
A DTHS 
control problem $({\cal H}, I, G)$ is defined as the LTS
control problem ($\mbox{\rm LTS}(\cal H)$, $I$, $G$).
%
%
%
%
To accommodate quantization errors,
always present in software based controllers, it is useful
to relax the notion of control solution by tolerating
an (arbitrarily small) error $\varepsilon$ on the continuous variables.
This leads to the definition of $\varepsilon$-solution.
%
%
Let $\varepsilon$ be a nonnegative real number, 
$W\subseteq \R^n\times \Z^m$. 
The $\varepsilon$-\emph{relaxation} of $W$ is the set
(\emph{ball} of radius $\varepsilon$)
${\cal B}_{\varepsilon}(W)$
= \{($z_1, \ldots z_n$, $q_1, \ldots q_m$) 
$|$ $\exists (x_1, \ldots, x_n, q_1, \ldots q_m) \;
\in \; W$ and $\forall i \in [n] \; 
|z_i - x_i| \leq \varepsilon$\}.


\begin{definition}
  \label{def:dtlhs-ctr-prb}
  Let $({\cal H}, I, G)$ be a DTHS control problem and $\varepsilon$
be a nonnegative real number.
An $\varepsilon$ solution to $({\cal H}, I, G)$ is a 
solution to the LTS control problem 
$(\mathrm{LTS}({\cal H}), I, {\cal B}_{\varepsilon}(G))$. 
\end{definition}

\begin{example}
\label{ex:ctr-dths}
Let T be a positive constant (sampling time).
We define the DTHS ${\cal H}$ $=$ $(\{x\},\{u\},$ $\varnothing$, $N)$ where 
$x$ is a continuous variable, $u$ is a boolean variable, and
$N(x, u, x')$ $\equiv$
$[\overline{u} \rightarrow x' = x+(\frac{5}{4}-x)T] 
\land
[u \rightarrow x' = x+ (x - \frac{3}{2})T]$.
Let ${\cal P}$ = (${\cal H}$, $I$, $G$) be a control problem, where 
$I\equiv -2\leq x\leq 2.5$, and $G\equiv x=0$.
A controller may drive the system near enough to the goal $x=0$, 
by enabling a suitable action in such 
a way that $x'<x$ when $x>0$ and $x'>x$ when $x<0$.
If the sampling time $T$ is small enough with respect to $\varepsilon$
(for example $T<\frac{\varepsilon}{10})$, 
the controller:
$
K(x,u)=(-2\leq x\leq 0\;\land\; \overline{u}) \; \lor \; 
(0\leq x\leq \frac{11}{8} \;\land\; u) \; \lor \;
(\frac{11}{8}\leq x\leq 2.5 \; \land\; \overline{u})
$ 
is an $\varepsilon$ solution to $({\cal H}, I, G)$.
Observe that, that any controller $K'$ such that $K'(\frac{5}{4},0)$ holds is not 
a solution, because since $N(\frac{5}{4},0,\frac{5}{4})$ holds, 
the closed loop system ${\cal H}^{(K)}$ may loop forever along 
the path $\frac{5}{4},0,\frac{5}{4},0\ldots$.
\end{example}

\begin{example}
\label{ex:pendulum-goal}
The typical goal for the inverted pendulum in Example~\ref{ex:dths} is to turn the pendulum steady 
to the upright position, starting from any possible initial position, within a given speed interval. In our experiments, the goal region is defined by the predicate $G(X)\equiv (-\rho\leq x_1\leq\rho)\,\land\,(-\rho\leq x_2\leq\rho)$, where $\rho\in\{0.05, 0.1\}$, and the initial 
region is defined by the predicate $I(X)\equiv(-\pi\leq x_1\leq \pi)\,\land\,(-4\leq x_2\leq 4$). 
\end{example}




\subsection{Quantized Control Problem}
\label{sec:qcp}
In order to manage real variables, in classical control theory the
concept of {\em quantization} is introduced (e.g., see
\cite{quantized-ctr-tac05}). Quantization is the process of
approximating a continuous interval by a set of integer values. 
%
In the following we formally define a quantized feedback control
problem for DTHSs.

A {\em quantization function} $\gamma$
for a real interval $I=[a,b]$ is a non-decreasing   
function $\gamma:I\mapsto \Z$ s.t. $\gamma(I)$ is a bounded integer interval.
We will denote $\gamma(I)$ as $\hat{I}=[\gamma(a),\gamma(b)]$.
The \emph{quantization step} of $\gamma$, notation $\|\gamma\|$, 
is defined as ${\rm sup}\{ \; |w-z|
  \; | \; w, z \in I \land \gamma(w)=\gamma(z)\}$. 
For ease of notation, we extend quantizations to integer intervals,
by stipulating that in such a case the quantization function
is the identity function.

\begin{definition}
  Let ${\cal H} = (X, U, Y, N)$ be a DTHS, and 
  let $W=X\cup U\cup Y$.   
  A \emph{quantization} ${\cal Q}$ for $\cal H$
  is a pair $(A, \Gamma)$, where:
  \begin{itemize}
  \item

  $A$ is a 
  predicate over $W$ 
  that explicitely bounds each variable in $W$. 
  For each $w\in W$, we denote  
  with $A_w$ its {\em admissible region} and with $A_W$ $=$ $\prod_{w\in W} A_w$.
  
  \item 
  $\Gamma$ is a set of maps $\Gamma = \{\gamma_w$ $|$
  $w \in W$ and $\gamma_w$ is a 
  quantization function for $A_w\}$. 
  \end{itemize}
  Let $W = [w_1, \ldots w_k]$ 
  and $v = [v_1, \ldots v_k] \in A_{W}$. 
  We write $\Gamma(v)$ for the tuple $[\gamma_{w_1}(v_1),$ $ 
  \ldots,$ $\gamma_{w_k}(v_k)]$. 
  Finally, the \emph{quantization step} $\|\Gamma\|$ is 
  defined as ${\rm sup} \{ \; \|\gamma\| \; | \; \gamma \in
  \Gamma \}$.
\end{definition}
A control problem admits a \emph{quantized} solution if control
decisions can be made by just looking at quantized values. This enables
a software implementation for a controller.


\begin{definition}
  \label{def:qfc}
  Let ${\cal H} = (X, U, Y, N)$ be a DTHS,
  ${\cal Q}=(A,\Gamma)$ be a quantization for ${\cal H}$
  and ${\cal P} = ({\cal H}, I, G)$ be a DTHS control problem.
  A ${\cal Q}$ \emph{Quantized Feedback Control} (QFC) 
  solution to ${\cal P}$ is a
  $\|\Gamma\|$ 
  solution $K(x, u)$ to ${\cal P}$ such that
  $
  K(x, u) = \hat{K}(\Gamma(x), \Gamma(u))
  $
  where 
  $\hat{K} : \Gamma(A_{X}) \times \Gamma(A_{U})$ $\rightarrow$ $\B$.
\end{definition}



\begin{example}
\label{ex:q-ctr}
Let ${\cal P}$ be as in Example \ref{ex:ctr-dths}.
Let us consider the quantization $(A, \Gamma)$ where $A=I$
and $\Gamma$ = $\{\gamma_x\}$ where $\gamma_x(x)=\lfloor x\rfloor$.
The set $\Gamma(A_x)$ of quantized states is the integer interval $[-2,2]$. 
No ${\cal Q}$ QFC solution can exist, because defining 
both $\hat{K}(1,1)$ and $\hat{K}(1,0)$ allows infinite loops to be potentially executed in 
the closed loop system. Of course, the controller $K$ in Example \ref{ex:ctr-dths} can be 
obtained as a quantized controller decreasing the quantization step, for example by taking 
$\tilde{\Gamma}$ = $\{\tilde{\gamma}_x\}$ where $\tilde{\gamma}_x(x)=\lfloor 8x\rfloor$. 
\end{example}

 
\section{DTLHS overapproximation of DTHSs}
In~\cite{qks-cav10}, we presented the tool QKS that given a DTLHS control problem 
${\cal P}=({\cal H}, I, G)$ and a quantization schema as input, yields as output control software 
implementing a most general optimal quantized controller for ${\cal P}$, 
whenever a sufficient condition holds.
In this section we show how a DTHS ${\cal H}$ 
can be overapproximate by a DTLHS ${\cal L}_{\cal H}$, in such a way that 
$\mbox{LTS}({\cal H})\sqsubseteq \mbox{LTS}({\cal L}_{\cal H})$. 
The following theorem ensures that controllers for ${\cal L}_{\cal H}$ are 
guaranteed to be controllers for ${\cal H}$.

\subsection{DTHS linearization}
Let $C(V)$, with $V \subseteq X \cup U \cup Y \cup X'$, be a constraint in $N$ that contains a nonlinear function as 
a subterm. Then $C(V)$ has the shape $f(R,W)+E(V)\leq b$, 
where $R\subseteq V^r$ is a set of $n$ real variables $\{r_1,\ldots, r_n\}$, 
and $W\subseteq V^d$ is a set of discrete variables.
For each $w\in {\cal D}_W$, we define the function $f_w(R)$ 
obtained from $f$, by instanciating discrete variables with $w$, i.e
$f_w(R)=f(R,w)$.
Then $C(V)$ is equivalent to the conjunctive predicate 
$\bigwedge_{w\in {\cal D}_W} [f_w(R)+E(V)\leq b]$.
In order to make the overapproximation tighter, we partition the domain ${\cal D}_R$ 
of each function $f_w(R)$ into 
$m$ hyperintervals $I_1, I_2 \ldots I_m$, where 
$I_i=\Pi_{j\in [n]} [a_{j}^{i}, b_{j}^{i}]$. 
In the following $R\in I_i$ will denote the conjunctive predicate 
$\bigwedge_{j\in [n]} a_{j}^{i}\leq r_j\leq b_{j}^{i}$.


Let $f_{w,i}^+(R)$ and $f_{w,i}^-(R)$ be over- and under- linear approximations
of $f_w(R)$ over the hyperinterval $I_i$, i.e. such that $R\in I_i$ implies $f_{w,i}^-(R)\leq
f_w(R)\leq f_{w,i}^+(R)$.
Taking $|{\cal D}_W|\times n$ fresh continuous variables
$Y=\{y_{w,i}\}_{w\in{{\cal D}_W}, i\in[n]}$, we define the conjunctive predicate
$\tilde{C}(V,Y)$:
\[
\begin{array}{l}
\bigwedge_{w\in {\cal D}_W} \bigwedge_{i\in[m]} [y_{w,i}+E(V) \leq b]\\
\hspace*{.3cm} \land \bigwedge_{w\in {\cal D}_W} [\bigvee_{i\in[m]} [R\in I_i \land f_{w,i}^-(R)\leq y_{w,i}\leq f_{w,i}^+(R)]]
\end{array}
\]

By introducing $|{\cal D}_W| \times n$ fresh boolean variables $Z=\{z_i\}_{w\in{{\cal D}_W}, i\in[n]}$, $\tilde{C}(V,Y)$ can be translated into the following equisatisfiable conjunctive predicate
$\bar{C}(V,Y,Z)$:
\[
\begin{array}{l}
\bigwedge_{w\in {\cal D}_W} \bigwedge_{i\in[m]} [y_{w,i}+E(V) \leq b]\\
\hspace*{.3cm} \land \bigwedge_{w\in {\cal D}_W} \bigwedge_{i\in[m]} z_{w,i}\rightarrow f_{w,i}^-(R)\leq y_{w,i}\leq f_{w,i}^+(R)\\
\land \bigwedge_{w\in {\cal D}_W} \bigwedge_{i\in[m]} z_{w,i}\rightarrow R\in I_i 
\land \bigwedge_{w\in {\cal D}_W} \sum_{i\in[m]} z_{w,i} \geq 1
\end{array}
\]

As a result, this transformation eliminates a nonlinear subexpression of a constraint $C(V)$ 
and yields a constraint $\bar{C}(V,Y,Z)$ such that $\exists Y,Z [\bar{C}(V,Y,Z)\Rightarrow C(V)]$. 
Given a DTHS ${\cal H}=(X, U, Y, N)$, without loss of generality, we may suppose 
that the transition relation $N$ is a conjunction $\bigwedge_{i\in[m]} C_i(X,U,Y,X')$ 
of constraints. 
By applying the above transformation to each nonlinear subexpressions occurring in $N$,  
we obtain a conjunction of linear constraints $\bar{N}\equiv\bigwedge_{i\in[\bar{m}]} \bar{C}_i(X,U,\bar{Y},X')$, such that $\bar{N}\Rightarrow N$. Hence, starting from a DTHS ${\cal H}$, we find 
a DTLHS ${\cal L}_{\cal H}=(X, U, \bar{Y}, \bar{N})$, whose dynamics overapproximate the dynamics of ${\cal H}$.

\begin{theorem}
Let ${\cal H}=(X, U, Y, N)$ be a DTHS and let ${\cal L}_{\cal H}$ be its linearization.
Then we have that  
$\mbox{LTS}({\cal H})\sqsubseteq \mbox{LTS}({\cal L}_{\cal H})$.
\end{theorem} 

\begin{theorem}
\label{thm:ctr-refinement}
Let ${\cal S}_1=(S, A, T_1)$ and ${\cal S}_2=(S, A, T_2)$ be two LTSs, 
and let $K$ be a solution for the LTS control problem
$({\cal S}_2, I, G)$.
If ${\cal S}_1$ refines ${\cal S}_2$ and for all $s\in S$ $\mbox{Adm}({\cal S}_1,s)=\mbox{Adm}({\cal S}_1,s)$, then $K$ is a solution also for $({\cal S}_1, I, G)$. 
\end{theorem}
\begin{proof}
(Sketch) 
The proof is by induction on $n=J_{\rm strong}({\cal S}^{(K)}_2,G,s)$. 
If $n=1$ and $K(s,a)$, then ${\rm Img}({\cal S}_2,s,a)\subseteq G$. 
Since ${\cal S}_1\sqsubseteq {\cal S}_2$, we also have that 
${\rm Img}({\cal S}_1,s,a)\subseteq {\rm Img}({\cal S}_2,s,a)\subseteq G$. 
Moreover, $\mbox{Adm}({\cal S}_1,s)=\mbox{Adm}({\cal S}_2,s)$ implies that 
there exists at least a transition of the shape $T_1(s,a,s')$ with $s'\in G$ and thus 
$J_{\rm strong}({\cal S}^{(K)}_1,G,s)=1$ too. 
This implies that $\{s~|~J_{\rm strong}({\cal S}^{(K)}_1,G,s)=1\}=
\{s~|~J_{\rm strong}({\cal S}^{(K)}_2,G,s)=1\}$.
The inductive step is similar, by substituting $G$ with the set of states 
$\{s~|~J_{\rm strong}({\cal S}_2,G,s)=n-1\}$. 
\end{proof}

\begin{corollary}
\label{cor:real-result}
Let ${\cal H}=(X, U, Y, N)$ be a DTHS and let ${\cal L}_{\cal H}$ be its linearization.
Let $K$ be a solution for the DTLHS control problem
$({\cal L}_{\cal H}, I, G)$. Then $K$ is a solution also for the DTHS control problem $({\cal H}, I, G)$. 
\end{corollary}

%
%
%

\begin{example}
\label{ex:linearpendulum}

\begin{figure}[tb!]
  \centering
  \includegraphics[width=.8\columnwidth]{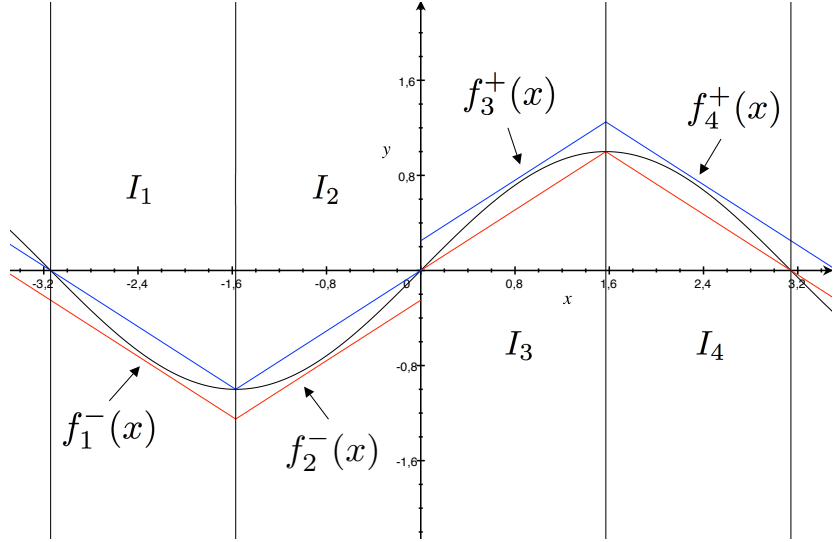}
  \caption{Linearization of $\sin x$ in $[-\pi, \pi]$.}
  \label{fig:linsin}
\vspace*{-0.6cm}
\end{figure}

The DTHS ${\cal H}=(X,U,\varnothing,N)$ model for the inverted pendulum in
Ex.~\ref{ex:dths} contains the nonlinear function $\sin x_1$. We define the
linearization ${\cal L}_{\cal H} = (X, U, Y, \tilde{N})$ as follows. In order to
exploit sinus periodicity, we consider the equation $x_1 = 2 \pi y_{k} +
y_{\alpha}$, where $y_{k}$ represents the period in which $x_1$ lies and
$y_{\alpha} \in [-\pi, \pi]$ represents the actual $x_1$ inside a given period.

This allows us to apply our linearization to $y_{\alpha}$ $\in$ $[-\pi, \pi]$ 
only. We partition the interval $[-\pi, \pi]$ into four sub-intervals $I_1$,
$I_2$, $I_3$, $I_4$ as shown in Fig.~\ref{fig:linsin}. For $y_{\alpha} \in I_1 =
[-\pi, -\frac{\pi}{2}]$ we define $f_1^{+}(y_{\alpha})$ as the line passing
through points $(-\pi, \sin (-\pi))$ and $(-\frac{\pi}{2}, \sin
(-\frac{\pi}{2}))$, i.e. $f_1^{+}(y_{\alpha})$ $=$ $-0.6369 y_{\alpha} + 2$.
Moreover, we define $f_1^{-}(y_{\alpha})$ as the line which is tangent to the
curve $\sin y_{\alpha}$ at $I_1$ medium point, i.e. $f_1^{-}(y_{\alpha})$ $=$
$0.7073 (y_{\alpha} + 0.785) - 0.7068$. Functions $f_2^{\pm}$, $f_3^{\pm}$ and
$f_4^{\pm}$ are obtained analogously.

Finally, we have that $Y = Y^d \cup Y^r = \{y_k, y_q, z_1, z_2, z_3, z_4\} \cup
\{y_{\alpha}\}$ and $\tilde{N} \equiv (x'_1 = x_1 + 2\pi y_q + T x_2)\,\land\,(x'_2 = x_2 +
T {g \over l} y_{\alpha} + T {1 \over m l^2} F u) \land x_1 = 2 \pi y_{k} + y_{\alpha}
\land \bigwedge_{i = 1}^4 z_i \rightarrow f^-_i \leq y_{\alpha} \leq f^+_i
\land \bigwedge_{i = 1}^4 z_i \rightarrow x_1 \in I_i \land \sum_{i=1}^4 z_i
\geq 1$.

%

\end{example}

\subsection{Linearization: a systematic approach}
When nonlinear subexpressions are ${\cal C}^2$ functions,  
a systematic approach to compute linear overapproximations of a DTHS makes use  
of Taylor polinomial of degree 1 as piecewise affine functions that  
over- and under-approximate the value of a ${\cal C}^2$ function. 
Let $f(x)$ be a ${\cal C}^2$ function of $n$ real variables over a given interval $I$.
By Taylor's theorem, we may derive {\em linear} 
under- and over-approximations for $f(x)$ around a given point $x_0 \in I$
as follows. Namely, we have that exists $t \in [0, 1]$ such that 
$f(x)=f(x_0)+\medtriangledown f(x_0)(x-x_0)+\frac{1}{2}(x - x_0)^TH(x + t(x - x_0)) (x - x_0)$, 
being $H$ the Hessian matrix of $f$.
If we know two real numbers $m$ and $M$ that are the minimum and 
the maximum value of $\frac{1}{2}(x - x_0)^TH(x + t(x - x_0)) (x - x_0)$, in a given 
interval around $x_0$. In this case we can choose 
$f^+(x)=f(x_0)+\medtriangledown f(x_0)(x-x_0)+M$ and 
$f^-(x)=f(x_0)+\medtriangledown f(x_0)(x-x_0)+m$.


        
  \begin{figure*}
  \centering
  \begin{tabular}{ccc}
  \begin{minipage}{0.31\textwidth}
  \includegraphics[width=\columnwidth]{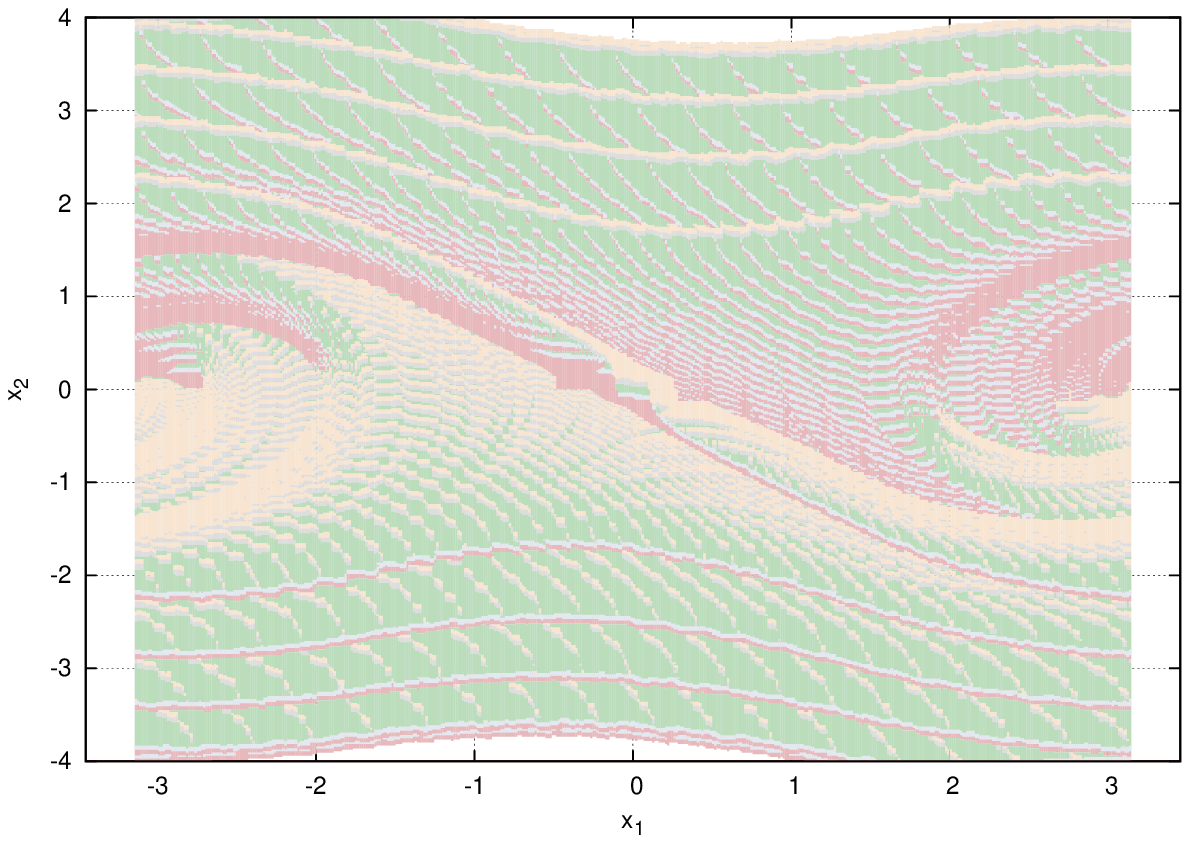}
  \caption{Controllable region for $F=0.5$, $T=0.1$, and $b=9$.}
  \label{fig:controllable-region-05}
  \end{minipage}
&
\begin{minipage}{0.31\textwidth}
  \includegraphics[width=\columnwidth]{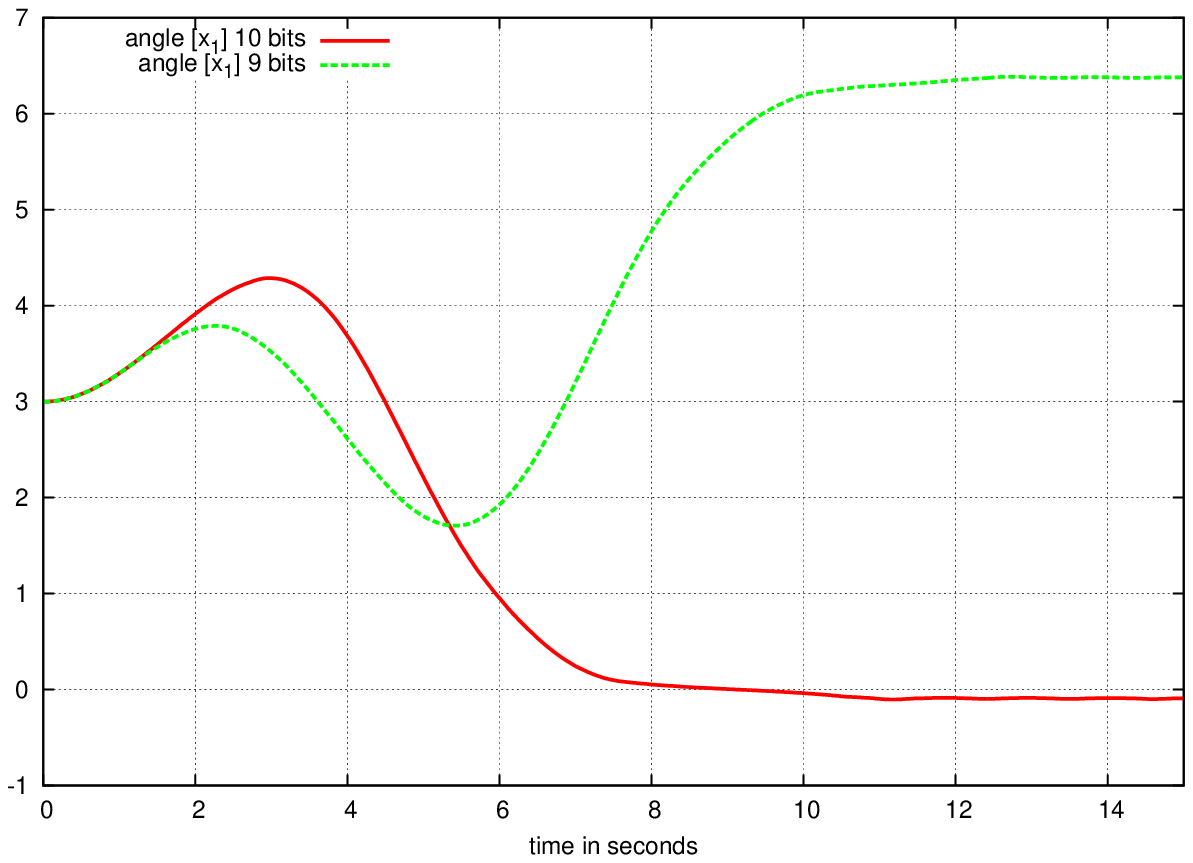}
  \caption{Trajectories for ${\cal H}^{(K_{0.5}^{(9)})}$ and ${\cal H}^{(K_{0.5}^{(10)})}$ starting from $(x_1, x_2) = (\pi, 0)$.} 
  \label{fig:trajectories-05}
\end{minipage}
&
\begin{minipage}{0.31\textwidth}
  \includegraphics[width=\columnwidth]{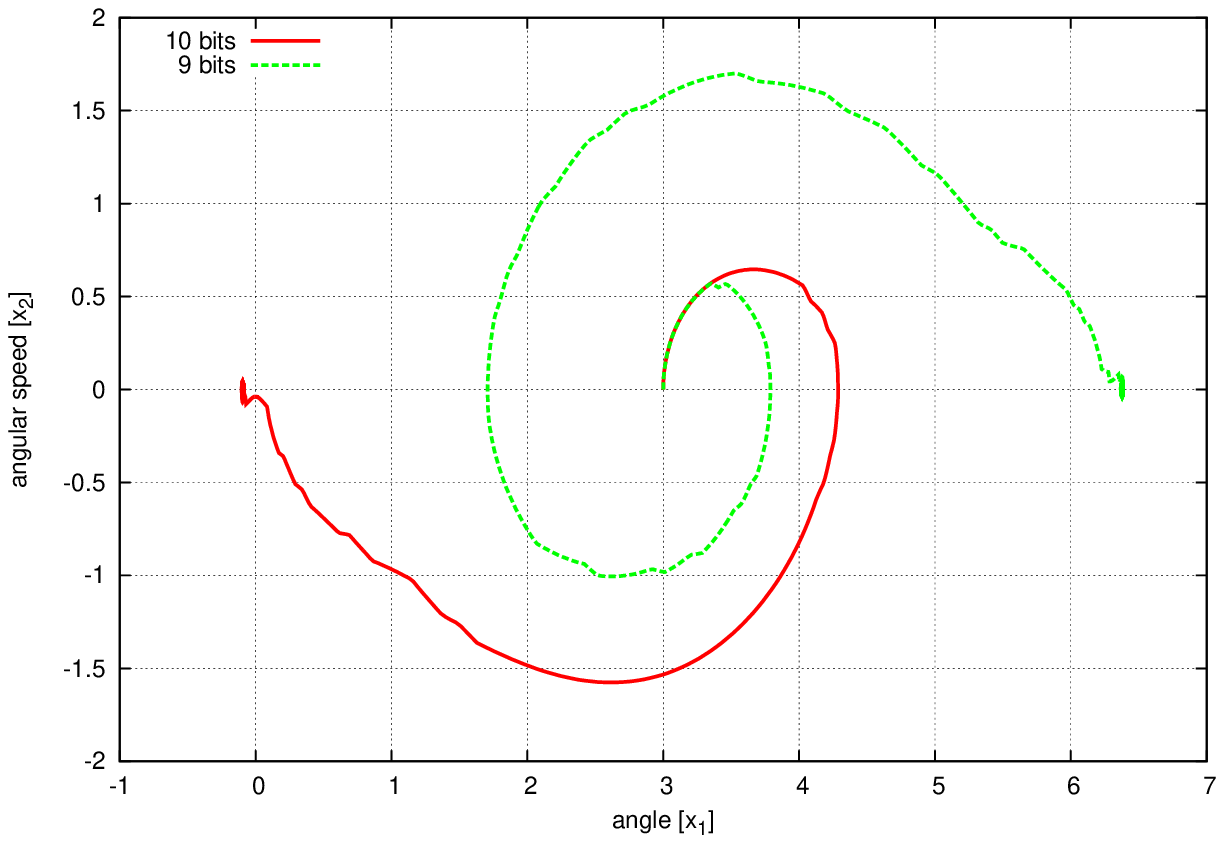}
  \caption{Same trajectories of Fig.~\ref{fig:trajectories-05} in the phases space.} 
  \label{fig:phase-plane-05}
\end{minipage}
\end{tabular}
\vspace*{-0.5cm}
\end{figure*}

\vspace*{-0.1cm}
\section{Experimental Results}
\vspace*{-0.1cm}
In this section we present our experiments that aim at
evaluating effectiveness of our linearization technique. 

\vspace*{-0.1cm}
\subsection{Experimental Settings}
\label{expres.tex.setting}

We present experimental results obtained by using \qks~\cite{qks-cav10} on the
inverted pendulum described in Example~\ref{ex:dths}. In order to let \qks\
handle such a case study, we linearize the DTHS ${\cal H}$ in
Example~\ref{ex:dths} with the DTLHS ${\cal L}_{\cal H}$ of
Example~\ref{ex:linearpendulum}. In all our experiments, as in~\cite{KB94} we
set  parameters $l$ and $m$ in such a way that  $\frac{g}{l}=1$ (i.e. $l=g$) and
$\frac{1}{ml^2}=1$ (i.e. $m=\frac{1}{l^2}$). As for the quantization, we set $A_{x_1} =
[-1.1\pi, 1.1\pi]$ and $A_{x_2} = [-4, 4]$, and we define $A = A_{x_1} \times
A_{x_2} \times A_u$. Moreover, we use uniform quantization functions dividing the domain of each state
variable ($x_{1}, x_2$) into $2^b$ equal intervals, where $b$ is the number of
bits used by AD conversion. The resulting quantization is ${\cal Q}_b = (A,
\Gamma_b)$, with $\|\Gamma_b\| =  \frac{8}{2^b}$. 
Since we have two quantized
variables ($x_1, x_2$) each one with $b$ bits,  the number of quantized
(abstract) states 
is exactly $2^{2b}$.   
Finally, the initial region $I$ and goal region $G$ are as
in Ex.~\ref{ex:pendulum-goal}, thus the DTHS [DTLHS] control problem we consider
is  $P$ = (${\cal H}$, $I$, $G$) [(${\cal L}_{\cal H}$, $I$, $G$)].

%
%
%
%
%
%
%


We run \qks\ for different values of the remaining parameters, i.e. $F$ (force
intensity), $\rho$ (goal tolerance), $T$ (sampling time), and $b$ (number of
bits of AD).  For each of such experiments, \qks\ outputs a control software $K$
in C language. In the following, we sometimes make explicit the dependence on
$F$ and $b$ by writing $K_{F}^{(b)}$. In order to evaluate performance of $K$,
we use an  {\em inverted pendulum  simulator} written in C. The simulator
computes the next state by using Eq.~(\ref{eq:pendmotion2}) of
Ex.~\ref{ex:dths},   thus simulating a path of ${\cal H}^{(K)}$. Such simulator
also implements the following features:

\begin{itemize}

	\item random disturbances (up to 4\%) in the next state computation
	are introduced, in order to
assess $K$ robustness w.r.t. non-modelled disturbances;

	\item Eq.~(\ref{eq:pendmotion2}) is translated into the discrete time
version by means of a simulation time step $T_s$ much smaller than the sampling
time $T$  used in ${\cal H}$ (and ${\cal L}_{\cal H}$). Namely, $T_s = 10^{-6}$
seconds, whilst $T = 0.01$ or $T = 0.1$ seconds. This allows us to have a more
accurated simulation. Accordingly, $K$ is called each $10^4$ (or $10^5$)
simulation steps of ${\cal H}$. When $K$ is not called, the last chosen action
is selected again ({\em sampling and holding}). 

\end{itemize}

All experiments have been carried out on an Intel(R) Xeon(R) CPU @ 2.27GHz, 
with 23GiB of RAM, Kernel: Linux 2.6.32-5-686-bigmem, distribution Debian
GNU/Linux 6.0.3 (squeeze).

\subsection{Underactuated Inverted Pendulum ($F = 0.5$)}

In order to stabilize an {\em underactuated} inverted pendulum (i.e. when $F<
1$)  from the hanging position to the uprigth position, a controller  needs to
find a non obvius strategy that consists of swinging the pendulum once or more
times   to gain enough momentum. We show that \qks\ is able to synthesize such a
controller by running it on ${\cal L}_{\cal H}$ where $F=0.5$ (note that
in~\cite{KB94} $F=0.7$). Results are in Tab.~\ref{expres.table.tex}, where each
row corresponds to a \qks\ run. Columns meaning in Tab.~\ref{expres.table.tex}
are as follows. Columns $b$, $T$ and $\rho$ show the corresponding inverted
pendulum parameters. Column $|K|$ shows the size of the C code for
$K_{0.5}^{(b)}$. Finally, columns {\bf CPU} and {\bf RAM} show the computation
time (in seconds) and RAM usage (in KB) needed by \qks\ to
synthesize $K_{0.5}^{(b)}$.

As for $K_{0.5}^{(b)}$ performance, it is easy to show that by reducing the
sampling time $T$ and the quantization step (i.e. increasing $b$), we increase
the quality of $K_{0.5}^{(b)}$ in terms of ripple, set-up time and coverage.  In
fact, Fig.~\ref{fig:trajectories-05} shows the simulations of ${\cal
H}^{(K_{0.5}^{(9)})}$ and ${\cal H}^{(K_{0.5}^{(10)})}$.  As we can see,
$K_{0.5}^{(10)}$ drives the system to the goal with a smarter trajectory,  with
one swing only. This have a significant impact on the set-up time  (the system
stabilizes after about $8$ seconds when controlled by $K_{0.5}^{(10)}$ instead of about  $10$
seconds required when controlled by $K_{0.5}^{(9)}$).
Fig.~\ref{fig:controllable-region-05} shows that the {\em controllable region}
of $K_{0.5}^{(9)}$ (i.e., ${\rm dom}(K_{0.5}^{(9)})$) covers almost all
states in the admissible region that we consider. Different colors mean different 
set of actions enabled by the controller. We observe that the mgo solution 
enables more than one action in a significant portion of the controllable region. 
The control software, however, is generated in such a way that one action is chosen 
in each state.
Finally,
Fig.~\ref{fig:ripple-05} shows the ripple of $x_1$ for ${\cal H}^{(K_{0.5}^{(10)})}$
inside the goal. Note that such ripple is very low (0.018 radiants).


\begin{figure*}
  \centering
  \begin{tabular}{ccc}
\begin{minipage}{0.31\textwidth}
  \includegraphics[width=\columnwidth]{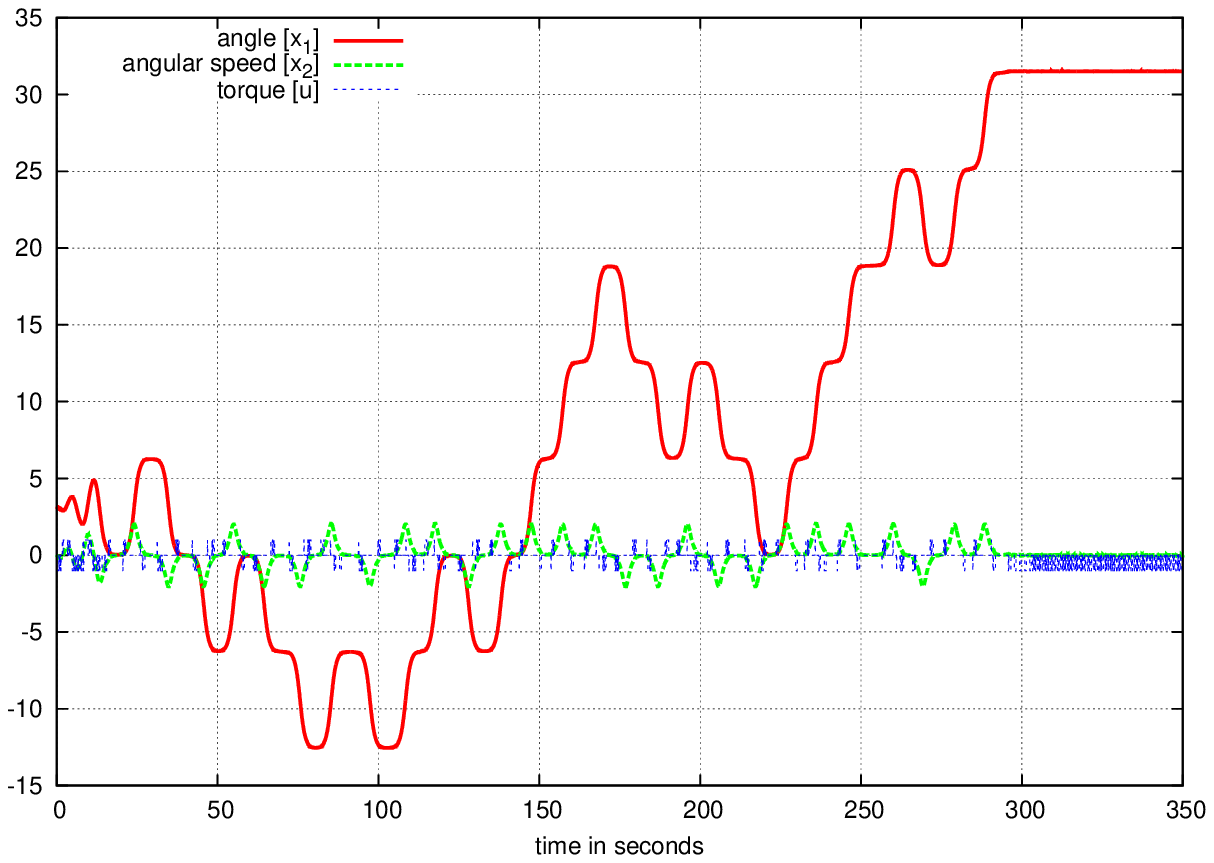}
  \caption{Simulation for ${\cal H}^{(K_{2}^{(11)})}$ starting from $(x_1, x_2) = (\pi, 0)$.}
  \label{fig:traj-03}
\end{minipage}
&
  \begin{minipage}{0.31\textwidth}
  \includegraphics[width=\columnwidth]{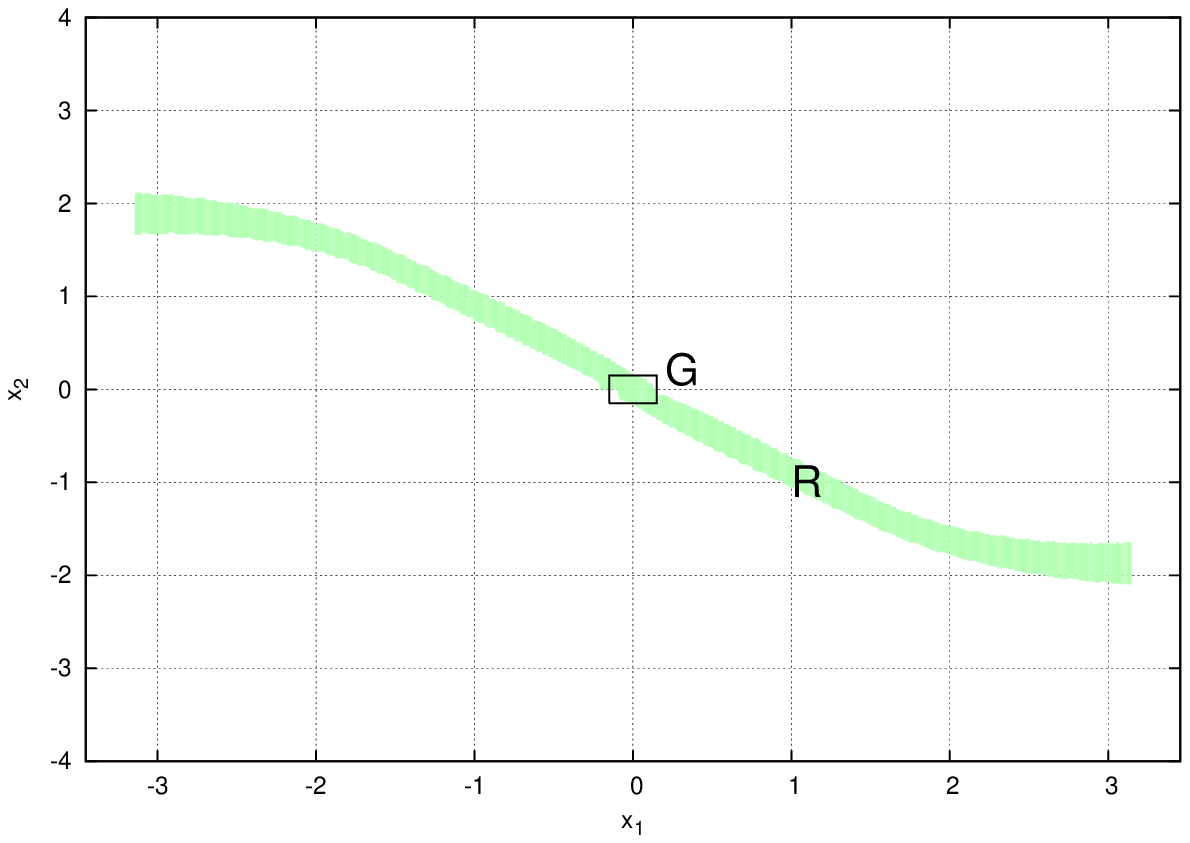}
  \caption{States turned directly to the goal with $F =0.3$.}
  \label{fig:layer-u0.3}
  \end{minipage}
&
\begin{minipage}{0.31\textwidth}
  \includegraphics[width=\columnwidth]{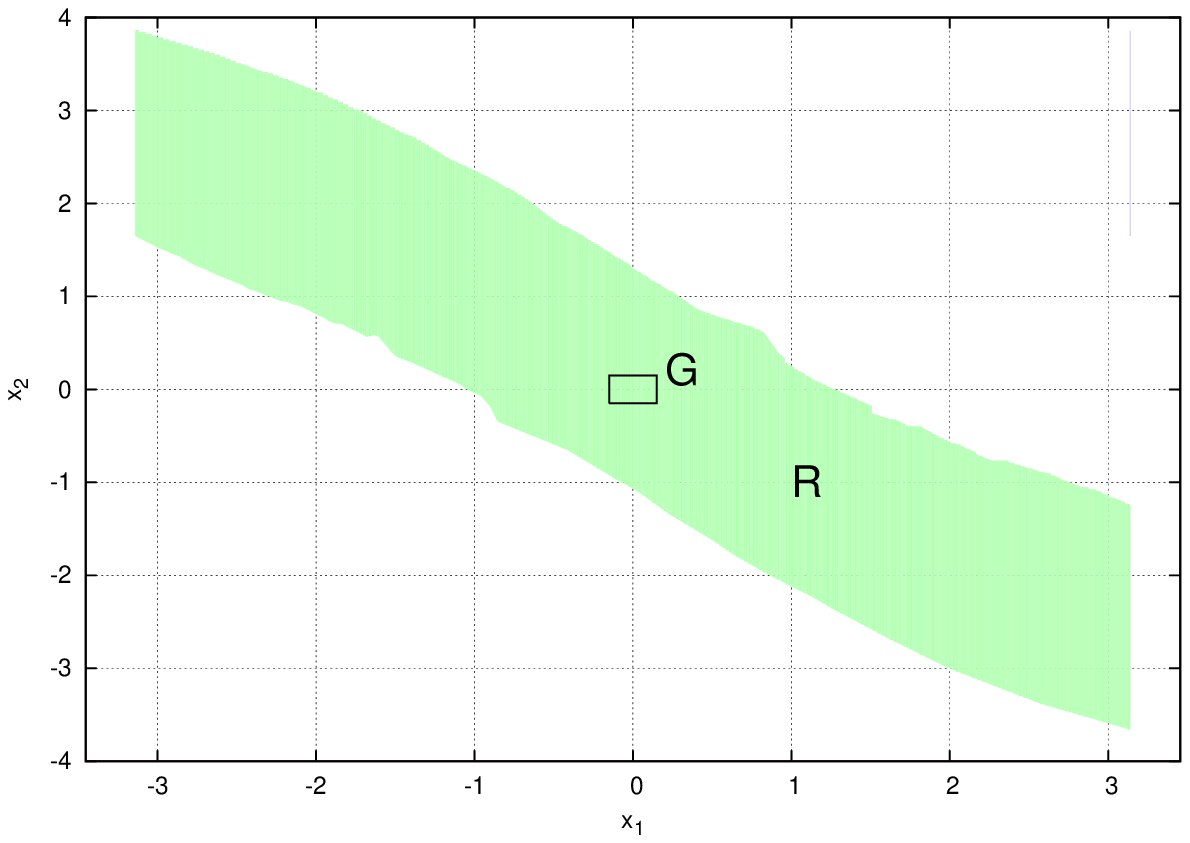}
  \caption{States turned directly to the goal with $F=2$.}
  \label{fig:layer-u2}
\end{minipage}
\end{tabular}
\vspace*{-0.5cm}
\end{figure*}

\subsection{Very Underactuated Inverted Pendulum  ($F = 0.3$)}


We succeeded to find controllers for the inverted pendulum for any value of $F$
down to $0.3$, with $T=0.1$ seconds and $\rho=0.1$. However, simulations show
that the behaviour of the resulting closed loop system is somewhat puzzling. As
it is shown in Fig.~\ref{fig:traj-03} for ${\cal H}^{(K_{0.3}^{(11)})}$, after
three swings the pendulum is correctly driven to the goal, but at that point
the controller is not able to maintain the plant inside the goal. In fact,
the controller let the pendulum fall and makes it do a complete round in order
to reach again the upright position.  This behaviour is repeated 27 times,
before the $K_{0.3}^{(11)}$ makes pendulum stabilize into the goal region.

As already noted in \cite{KB94}, all controllers for underactuated pendulum  use two very
different strategies to stabilize the system  depending on the initial state.
When the angle is positive and the speed is negative (and in a suitable  range
that depends on $F$), the controller turns directly the pendulum into the
upright position. Symmetrically, this also happens when the angle is negative
and the speed is positive.  Otherwise the controller let the pendulum fall down to gain
enough momentum (or to smoothly slow down it). Therefore, starting from very
near states may lead the system to follow  very different trajectories.     
Reducing $F$ squeezes the region of states from which the pendulum is directly
turned into  the upright position. As Fig.~\ref{fig:layer-u0.3} shows, when $F$
is equal to $0.3$,  we have a rather pathological situation:   the frontier
between the two strategies lies {\em inside} the goal region.  
The controller sometimes is unable to keep the system inside the goal, because disturbances 
introduced by the simulator make the system cross the frontier between the two strategies. 
When this frontier lies far enough from the goal  (see Fig.~\ref{fig:layer-u2} for the
case $F=2$), this phenomenon is essentially harmless and leads,  at worst, to
suboptimal strategies.



\begin{figure}
  \centering
  \begin{tabular}{ccc}
  \begin{minipage}{0.42\textwidth}
  \centering
  \includegraphics[width=\textwidth]{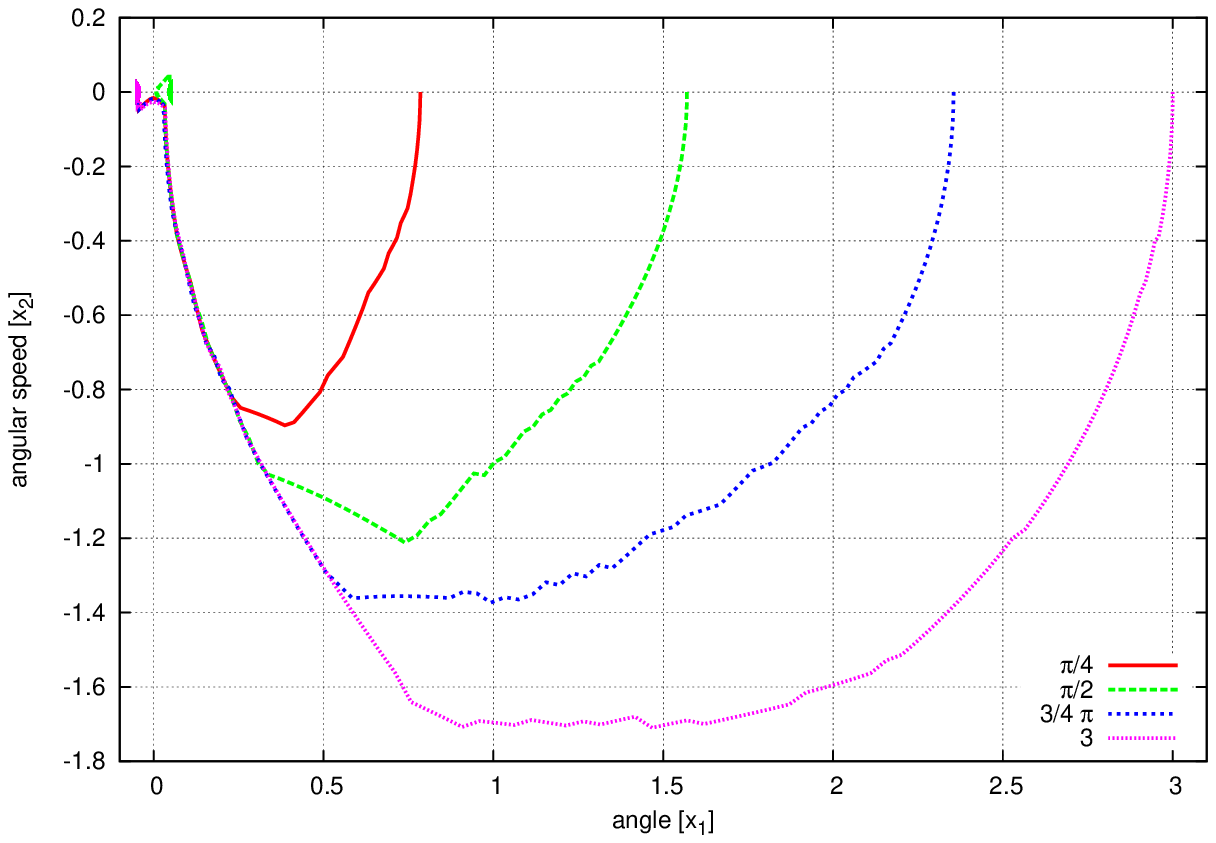}
  \caption{${\cal H}^{(K_{2}^{(11)})}$ phases space (starting from $x_1 \in \{{\pi \over 4}, {\pi \over
  2}, {3\pi \over 4}, 3\}$).} 
  \label{fig:mulpha}
  \end{minipage}
&
\begin{minipage}{0.42\textwidth}
  \includegraphics[width=\textwidth]{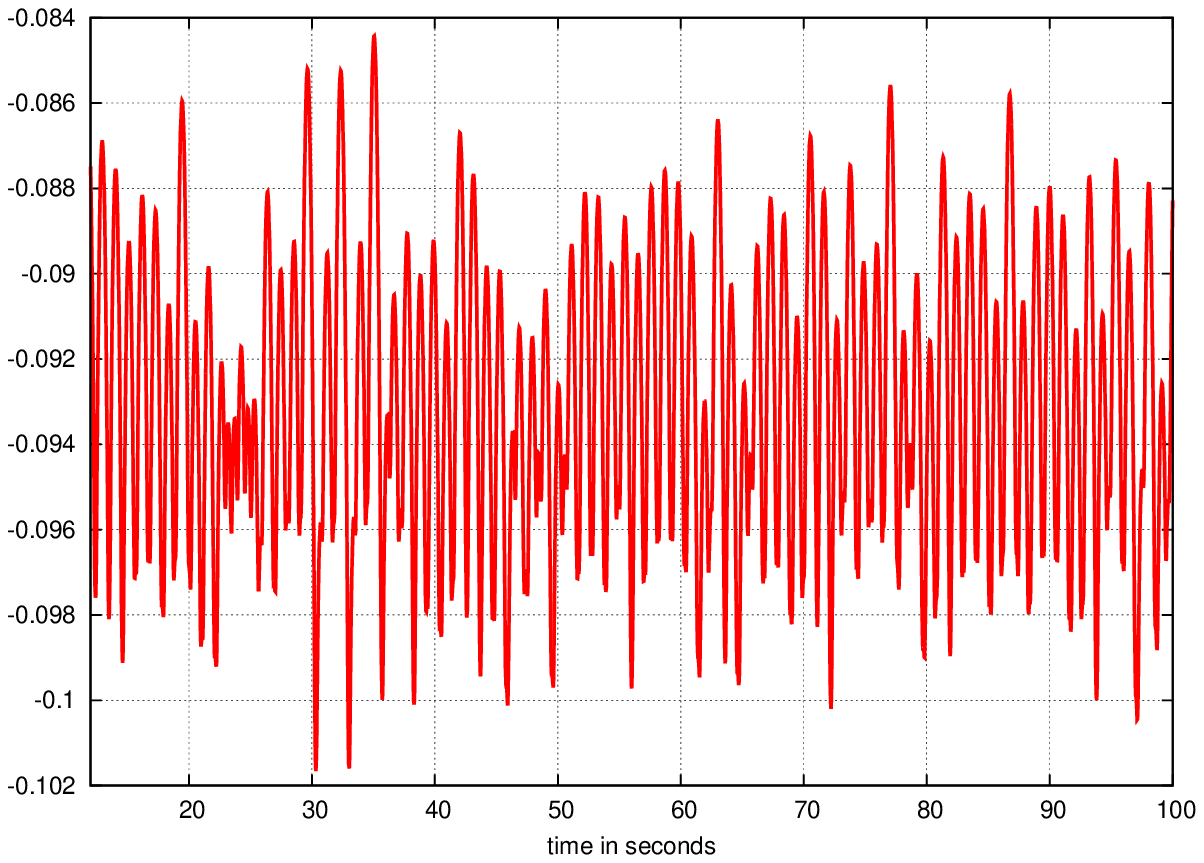}
  \caption{Ripple of $x_1$ for ${\cal H}^{(K_{0.5}^{(10)})}$.}
  \label{fig:ripple-05}
\end{minipage}
\end{tabular}
\end{figure}

\subsection{Overactuated Pendulum ($F=2$)}

When $F$ is greater than 1, finding a control strategy is less challenging.  It
is worth noting however that, even in this case,  our approach allows us to find
controllers that hardly can be  synthesized by means of traditional analytical
methods. In Fig.~\ref{fig:mulpha}, we show trajectories in the phases space of
${\cal H}^{(K_{2}^{(11)})}$ with $T=0.01$ seconds, $\rho = 0.05$,  and starting
values for $x_1$ are in $\{{\pi \over 4}, {\pi \over 2}, {3\pi \over 4}, 3\}$ and
$x_2 = 0$. ${\cal H}^{(K_{2}^{(11)})}$ follows highly non-smooth trajectories: 
$K_{2}^{(11)}$ drives the system along an optimal 
approach to the goal. Before joining this ideal path to the goal, the
controller,  in order to optimize the set up time, drives the system at the
maximum possible ``cruising'' speed  that allows the pendulum to be stopped in the goal. 
For higher values of $F$, this cruising speed is even higher.
%
%
%
%
%
%
%
%
%
%
%
%
    
\begin{table}[!tb]
  \centering
  \small

  \caption{Experimental Results for inverted pendulum with $F=0.5$.}\label{expres.table.tex}

  \begin{tabular}{ccc|ccc}
    \toprule
    $b$ & $T$ & $\rho$ & $|K|$ & CPU & MEM\\
    \midrule

8 & 0.1 & 0.1 & 2.73e+04 & 2.56e+03 & 7.72e+04 \\
9 & 0.1 & 0.1 & 5.94e+04 & 1.13e+04 & 1.10e+05 \\
10 & 0.1 & 0.1 & 1.27e+05 & 5.39e+04 & 1.97e+05 \\
11 & 0.01 & 0.05 & 4.12e+05 & 1.47e+05 & 2.94e+05 \\

    \bottomrule
  \end{tabular}

\end{table}

  \section{Conclusions}

We presented an automatic methodology to sinthesize  
control software for nonlinear Discrete Time Hybrid Systems. 
The control software is correct-by-construction with respect both 
System Level Formal Specifications of the closed loop system and Implementation Specification, 
namely the quantization schema. 
Our experimental results on the inverted pendulum benchmark show the effectiveness
of our approach and that we synthesize near optimal controllers that hardly can be 
designed by using traditional analytical methods of Control Engineering. 

The present work can be extended in several directions. 
First of all, it would be interesting to consider control synthesis of controllers 
that are optimal with respect a cost function given as input of the control problem, 
rather than simply time-optimal.
Another natural possible future research direction is to investigate
fully symbolic control software synthesis algorithms
based, for example, on efficient quantifier elimination procedures, in order 
to efficiently deal with Hybrid Systems with several continuous state variables.


\bibliographystyle{plain}
\small

\end{document}